\documentclass{article}
\usepackage{cite}

\usepackage[cmex10]{amsmath}
\usepackage{amssymb, amsthm}
\usepackage{graphicx}
\usepackage{longtable, booktabs}
\usepackage{cleveref}
\usepackage{wrapfig}
\usepackage{enumerate}

\theoremstyle{theorem}
\newtheorem{theorem}{Theorem}
\newtheorem{lemma}{Lemma}
\newtheorem{corollary}{Corollary}
\newtheorem{proposition}{Proposition}

\theoremstyle{definition}
\newtheorem{protocol}{Protocol}
\newtheorem{definition}{Definition}
\newtheorem{example}{Example}
\newtheorem{remark}{Remark}

\renewcommand{\le}{\leqslant}

\renewcommand{\sp}{{\rm sp}}
\newcommand{\herm}{{\rm H}}
\newcommand{\Q}{{\rm Q}}
\newcommand{\W}{{\rm W}}

\title{Improved User-Private Information Retrieval via Finite Geometry}

\author{Oliver W.~Gnilke\footnote{Aalto University, oliver.gnilke@aalto.fi} \and Marcus Greferath\footnote{Aalto University, marcus.greferath@aalto.fi} \and  Camilla Hollanti\footnote{Aalto University, camilla.hollanti@aalto.fi} \and Guillermo Nu{\~n}ez Ponasso\footnote{Aalto University, guillermo.nunezponasso@gmail.com} \and Padraig \'{O} Cath\'{a}in\footnote{Worcester Polytechnic Institute, pocathain@wpi.edu} \and Eric Swartz\footnote{College of William and Mary, easwartz@wm.edu}}

\begin{document}
\maketitle

\begin{abstract}
In a User-Private Information Retrieval (UPIR) scheme, a set of users collaborate to retrieve files from a database without revealing to observers which participant in the scheme requested the file. Protocols have been proposed based on pairwise balanced designs and symmetric designs. We
propose a new class of UPIR schemes based on generalised quadrangles (GQ).
\smallbreak

We prove that while the privacy of users in the previously proposed schemes could be compromised by a single user, the new GQ-UPIR schemes proposed in this paper maintain privacy with high probability even when up to $O(n^{1/4 - \epsilon})$ users collude, where $n$ is the total number of users in the scheme.
\end{abstract}

\section{Introduction}

\textit{Private Information Retrieval} (PIR) allows a user to retrieve information from a database without revealing which information was requested. A trivial solution is for the user to download all of the information on the database, but when the information is replicated in multiple locations, more efficient schemes are known \cite{CKGS98,BIKR02,DG16}.

A slightly different approach to the problem of private information retrieval attempts to hide the identity of the user downloading a file. One approach to this problem is Onion Routing \cite{GRS97}. An \textit{onion} is a recursively encrypted data packet, which encodes a path through a network of cooperating users. Each user in the path
removes the outermost layer of encryption, and forwards the onion to the next user on the path. The onion carries no identifying features, which would allow an observer to identify it with different outer layers. Anonymity is achieved in this system by choosing sufficiently long paths at random for the onions.

One disadvantage of onion routing is that the number of times a message is passed between users can be large. This results in a low throughput of data when bandwidth in the network is limited. \textit{User-Private Information Retrieval} (UPIR) is an approach
to private information retrieval in which the identities of users are hidden, but the number of times a message is forwarded through the network is tightly controlled.
To achieve privacy in a UPIR scheme, it is usual to place strong restrictions on which users can communicate with one another within the scheme.

\begin{definition}[cf. Section 2, \cite{SwansonStinsonI}] \label{def:UPIR}
	A \textit{UPIR system} consists of a bipartite graph $(\mathcal{U}\cup \mathcal{M},E)$, where $\mathcal{U}$ is the set of users and $\mathcal{M}$ is the set of message spaces. A user $u \in \mathcal{U}$ has access to a message space $M \in \mathcal{M}$ if $(u,M)$ is in the set of edges $E$. Furthermore it is assumed that all users have access to a database that evaluates queries.
\end{definition}

\begin{example}
	\Cref{fig:UPIR} shows a UPIR system with $5$ users and $3$ message spaces and the incidence matrix of the corresponding bipartite graph.
	\begin{figure}
		\centering
			\includegraphics[height=2cm]{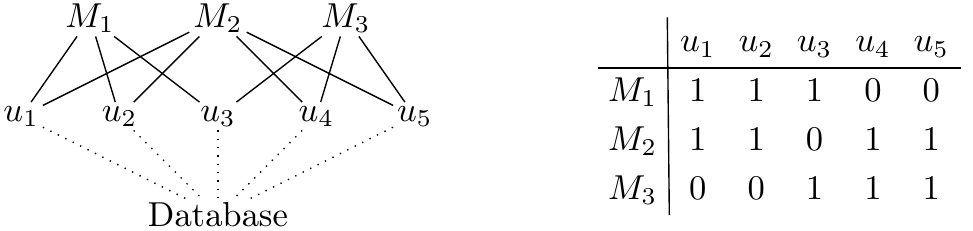}
			\caption{A visualisation of a UPIR system.}
			\label{fig:UPIR}
	\end{figure}
\end{example}

A common requirement in earlier work is that every pair of users share access to a common message space. Here, we require only that the bipartite graph underlying the UPIR scheme is connected; we say that such a UPIR system is \textit{connected}. If $u$ is incident with $M$, then user $u$ can both \textit{write} messages to $M$ and \textit{read} any messages written to $M$; if $u$ is not incident with $M$ then $u$ has no access to $M$. Users communicate within the scheme only by writing messages to one another in the message spaces. Any user may send queries to the database for evaluation. Users preserve their privacy by passing requests via the message spaces to other users to act as proxies for them.

In order for users to communicate using a UPIR system, a \textit{protocol} is required; one example is described explicitly below.
We refer to the combination of a UPIR system and a protocol as a \textit{UPIR scheme}. This distinction is helpful in illustrating the interactions between the combinatorics of the bipartite graph and the privacy  properties of the protocol.

\begin{protocol}\label{prot1}
	Let $(\mathcal{U} \cup \mathcal{M}, E)$ be a connected UPIR system.
	Suppose that user $u$ wishes to retrieve the response to the query $Q$ from the database.
	\begin{small}
	\begin{enumerate}
		\item User $u$ chooses a user $v$ uniformly at random from the set of all users.
		\item If $u = v$, $u$ requests $Q$ directly from the server, receiving response $R$.
		\item Otherwise, $u$ chooses a shortest path $(u, M_{1}, u_{1}, \dots, M_{n}, v)$ from $u$ to $v$ in the bipartite graph.
		\item User $u$ writes a \textit{request} $[ (u_{1}, M_{2}, \dots, M_{n}, v), Q]$ onto $M_{1}$.
		\item For $i = 1, 2, \dots, n-1$, user $u_{i}$ observes the request $[(u_{i}, M_{i+1}, \dots, M_{n}, v), Q]$ addressed to him in $M_{i}$. He writes a new request $[(u_{i+1}, M_{i+2}, \dots, M_{n}, v), Q]$ to $M_{i+1}$.
		\item When the message $[(v), Q]$ reaches message space $M_{n}$, user $v$ sees it and forwards $Q$ to the database. User $v$ writes the \textit{response} $R$ from the database as $[(v), Q, R]$ in $M_{n}$.
		\item User $u_j$, upon seeing the response $[(u_{j+1}, M_{j+2}, \ldots, v), Q, R]$ in $M_{j+1}$, writes the response    $[(u_{j},M_{j},u_{j+1},\ldots, v), Q, R]$ to $M_j$.
		\item User $u$ receives the response $R$ to his query after $u_1$ writes $[(u_{1}, M_{2}, \ldots, v), Q, R]$ to $M_{1}$.
	\end{enumerate}
\end{small}
\end{protocol}

\begin{remark}
Many variations of Protocol \ref{prot1} are possible, including randomising the path, and alterations to save used memory. Such changes do not alter the results in the following sections. Any user with access to a message space $M_{i}$ on the path can observe the request, the identity of the proxy and the response; but gains limited information about the source of the request. 
\end{remark}

Domingo-Ferrer, Stokes, Bras-Amor\'{o}s, and co-authors introduced UPIR systems and analysed
a protocol where users write queries to message spaces without specifying a proxy \cite{Domingo-Ferrer,StokesBrasAmores},
while a special case of the above protocol was developed by Swanson and Stinson
\cite{SwansonStinsonI, SwansonStinsonII}.
Both groups of authors worked on UPIR systems derived from highly structured set systems.
In particular, they required that every pair of users share a common message space,
in which case Protocol \ref{prot1} can be implemented so that every path has length at most $2$: any user can write requests directly to his chosen proxy.

Stokes and Bras-Amor\'{o}s considered the problem of constructing a UPIR system under the restrictions
that $\deg(M)$ is constant for all message spaces $M$, \cite{StokesBrasAmores}.
This requirement can be interpreted as balancing the load amongst message spaces. They also require that every pair of
users share precisely one message space. After rejecting some degenerate solutions in which message spaces
have size $1, 2, n-1$ or $n$; the authors are left with precisely the class of finite projective planes.
We recall that a projective plane is a combinatorial structure consisting of \textit{points} and \textit{blocks}
in which
\begin{enumerate}
    \item Every pair of points is contained in a unique block
    \item Every pair of blocks intersect in a unique point
    \item There exist $4$ points, no three contained in any block
\end{enumerate}
Projective planes have a rich theory playing an important role in combinatorics, geometry and algebra.
Inspired by Stokes and Bras-Amor\'{o}s, Swanson and Stinson analysed attacks on projective plane UPIR systems,
and proposed UPIR systems constructed from a broader class of block designs. In particular, they considered balanced incomplete block designs and pairwise
balanced designs (PBD). The monograph by Beth, Jungnickel and Lenz is a standard reference for design theory \cite{BJL}.

In the next section, we will show that observers in a UPIR system have an advantage in gathering information about users with whom
they share a message space. Motivated by this result, we consider the next obvious class of UPIR systems; those in which users are
separated by a path of length at most $2$. A natural class of examples are furnished by finite generalised quadrangles. Furthermore we consider a different protocol, based on onion routing, and prove that it aids privacy.

\section{Privacy in a UPIR system}

It is assumed throughout that the content of any message space is only available to the users who have access to the given message space as in Protocol \ref{def:UPIR}. An external eavesdropper, \emph{i.e.}, someone who is not a user in the system, can observe the requests made to the database, since these are not encrypted, but cannot read messages sent between users. Security in this setting has been studied previously in \cite{SwansonStinsonI} and their result forms the basis of any UPIR scheme.

\begin{definition}
A UPIR system is \textit{private with respect to external observers} if, for any
request $Q$ forwarded to the database by user $v$, we have that
\[ \mathbb{P}(u \textrm{ is the source of } Q \mid v \textrm{ is the proxy}) = \mathbb{P}(u \textrm{ is the source of } Q )\,. \]
\end{definition}

Swanson and Stinson have proved that the obvious strategy in which users
select proxies uniformly at random is sufficient for privacy against external
observers.

\begin{theorem}[Theorems 6.1, 6.2 \cite{SwansonStinsonI}]\label{thm:anon}
A connected UPIR scheme is private against external observers if each user chooses
proxies uniformly at random, and the proxies for distinct queries are chosen
independently.
\end{theorem}

Protocol \ref{prot1} can be implemented on any connected UPIR scheme;
and Theorem \ref{thm:anon} shows that the scheme is private against external observers.

More recent research on UPIR has aimed at preserving
privacy with respect to other users in the UPIR scheme, under the assumption that users
are \textit{honest but curious}. That is, they act according to Protocol \ref{prot1}, but they may attempt to determine the source of any queries that they observe.

Since the message space into which a request is written already reveals non-trivial
information about the source of the request, perfect privacy with respect to other
users is, in general, impossible. For example: in a PBD-UPIR scheme, it can be
inferred that the source of a request written to the message space $M$ is a user
with access to $M$.

First, we develop a criterion for judging whether a UPIR system is secure in terms of maintaining users' privacy.
Our analysis will be based on \textit{linked queries}, which are a series
of queries which are identifiable as coming from a single
source. These were first introduced by Swanson and Stinson, who
provided the example of a series of requests for information about a fixed,
obscure topic \cite{SwansonStinsonI}.

\begin{definition}
Let $C$ be a coalition of users, collaborating to identify the source of a series of linked queries.
Users $u$ and $v$ are \textit{pseudonymous} with respect to $C$ if for any message space $M$ to which $C$
has access, and for which $\mathbb{P}( v \textrm{ is source} \mid Q \in M) \neq 0$,
\[ \frac{ \mathbb{P}(u \textrm{ is the source of } Q \mid Q \in M) }{  \mathbb{P}( u \textrm{ is the source of } Q)} =
 \frac{ \mathbb{P}( v \textrm{ is the source of } Q \mid Q \in M) }{  \mathbb{P}( v \textrm{ is the source of } Q)} \,.\]
\end{definition}

We allow the possibility that a coalition has non-trivial prior information on the probability that user $u$
wishes to evaluate query $Q$. In our analysis we focus on the case where this information is limited, and users
cannot be identified by their queries alone. The next result follows directly from the definition of pseudonymity,
but we record it since we will have use for it in later sections.

\begin{lemma}
Pseudonymity with respect to the coalition $C$ is an equivalence relation on the users of a UPIR scheme.
\end{lemma}

The coalition $C$ can resolve the identity of a user $u$ if and only if $u$ belongs to a
pseudonymity class of size $1$ with respect to $C$. Users $u$ and $v$ maintain
pseudonymity with respect to $C$ after arbitrarily many requests have been observed
if and only if they lie in the same pseudonymity class. We propose the following definition for security.

\begin{definition}\label{defn:priv}
Let $(\mathcal{V}_i)$ be a family of UPIR schemes indexed by $i \in \mathbb{N}$, where the number of users in $\mathcal{V}_{i}$ is $n_{i}$.
We say that $\mathcal{V}_{i}$ is secure against coalitions of size $t$ if the pseudonymity relation of any coalition of size $t$
contains a \textit{giant component}, \emph{i.e.}, the union of all other components has size $O(n_{i}^{1-\epsilon})$ for some $\epsilon > 0$. The family $(\mathcal{V}_i)$ is \textit{secure} if for every $t$ there exists $N \in \mathbb{N}$ such that $\mathcal{V}_{i}$ is secure against coalitions of size $t$ for all $i \geq N$.
\end{definition}

Informally, we consider a UPIR scheme to be secure if any coalition of size at most $t$ can observe only a negligible portion of the system.
Equivalently, for any fixed coalition $C$ of bounded size a randomly chosen subset of users, of limited size, will be mutually pseudonymous with respect to $C$ with high probability. Our first result is that families of PBD-UPIR schemes are never secure.

\begin{theorem}\label{distance1}
In a PBD-UPIR scheme, a single eavesdropper can resolve the identity of any user who makes sufficiently many linked queries.
Equivalently, for any coalition of eavesdroppers, every pseudonymity class has size $1$.
\end{theorem}

\begin{proof}
Suppose that $u$ makes a series of linked queries.
An eavesdropper $c$ will observe a subset of these queries in the unique
message space $M$ shared by $c$ and $u$, and will never observe linked queries in
any other message space to which he has access. Since users do not write queries
addressed to themselves\footnote{In fact, Theorem \ref{thm:anon} forces
users to act as their own proxy equally often as any other user.
Even if users write requests addressed to themselves to message spaces, the
combinatorics of PBDs prohibit such requests from appearing with the same
frequency as those addressed to other users. Sources would hence still be resolvable,
though the result would now be probabilistic.}
 $c$ will be able to identify $u$ as soon as he has observed a linked query addressed to every other possible user in $M$. Provided that $u$ follows the requirements of Theorem \ref{thm:anon}, $c$ will observe the required queries with probability $1$. 
\end{proof}

In fact, a pair of collaborating users $c_{1}$ and $c_{2}$ can
identify $u$ far more quickly. If $c_{1}$ observes a query in
the message space $M_{1}$ and $c_{2}$ observes a linked
query in $M_{2}$, then the collaborators can conclude that
the source of the requests is a user in $M_{1} \cap M_{2}$.
But in a PBD-UPIR scheme, such a user is unique. This is
called an \textit{intersection attack} in \cite{SwansonStinsonI}.
Theorem \ref{distance1} can be easily modified
to identify all users in any UPIR scheme in which every pair
of users share at least one message space. In particular, all of the
UPIR schemes proposed by Swanson and Stinson to circumvent the
intersection attack are still vulnerable to Theorem \ref{distance1}; although it will take more linked queries to identify the source.

To protect against the attack outlined in \Cref{distance1} we suggest using a different incidence structure that we will introduce in the following section.

\section{Generalised Quadrangles}

In this section we introduce \textit{generalised quadrangles (GQ)}. For the sake of completeness, we include proofs of some well-known results, for further reading
see \cite{PayneThas}. We will show that the bipartite incidence graphs
of generalised quadrangles have diameter $4$. So in a UPIR scheme
derived from a GQ (GQ-UPIR scheme in short), a pair of users either shares
a message space, or there exists a third user sharing message spaces with
each of the first two. As a result, when users communicate along a shortest
path, a message is written to at most $2$ message spaces. In this section, we
use the usual language of incidence geometry; in a GQ-UPIR scheme, users
are labelled by \textit{points} and message spaces by \textit{blocks}.

\begin{definition}
A \textit{generalised quadrangle} is an incidence structure containing points and blocks in which two blocks can intersect in at most one point, and which satisfies the \textbf{GQ Axiom}: given any point $x$ and block $L$ that does not contain $x$, there is a unique point $x'$ in $L$ that shares a block with $x$.
\end{definition}

Even though we are dealing with an abstract incidence structure, there is a natural representation of this structure as a geometry.  It is traditional for the blocks in a generalised quadrangle to be referred to as \textit{lines}.  Indeed, a generalised quadrangle is so-named because there are no \textit{triangles} (three lines intersecting pairwise in three distinct points) but numerous \textit{quadrangles} in such a geometry.

 \begin{lemma}\label{triangle-free}
 There are no non-degenerate triangles in a generalised quadrangle, but any two non-collinear points are contained in a quadrangle.
 \end{lemma}

 \begin{proof}
 Recall that a \textit{triangle} is a triple of distinct lines $L_{1}, L_{2}, L_{3}$ with pair-wise non-empty intersections, say $x_{ij} \in L_{i} \cap L_{j}$.  Note that $x_{12}$ is collinear with both $x_{13}$ and $x_{23}$.  By the GQ Axiom, if $x_{12} \notin L_3$, then there exists a \textit{unique} point on $L_3$ collinear with $x_{12}$: in other words, $x_{13}$ and $x_{23}$ cannot both be collinear with $x_{12}$, and there are no triangles.

 On the other hand, consider two non-collinear points $x$ and $y$, and consider two lines $L_1$ and $L_2$ incident with $y$.  The point $x$ in not incident with either $L_1$ or $L_2$, and, by the GQ Axiom, there is a point $w$ on $L_1$ and a point $z$ on $L_2$ such that $x$ is collinear with both $w$ and $z$.  If the line incident with both $x$ and $z$ is $L_3$ and the line incident with both $x$ and $w$ is $L_4$, then the quadruple of distinct lines $L_1$, $L_2$, $L_3$, $L_4$ is the desired quadrangle.
 \end{proof}

For a point $x$ in a generalised quadrangle $\mathcal{Q}$ we write $B_{1}(x)$ for the set of points collinear with $x$. By convention, $x \notin B_{1}(x)$. Suppose that $y \notin \{x \} \cup B_{1}(x)$, and let $L$ be any line through $x$. By the GQ axiom, $y$ is collinear to a unique point on $L$, so $y$ is at distance $2$ from $x$, which we denote by $y \in B_{2}(x)$. In fact, since the choice of $L$ was arbitrary, we obtain a bijection: every line through $x$ intersects a unique line through $y$. Hence every point lies on a fixed number $t+1$ of lines. A similar argument shows that all lines have the same size, which we denote $s+1$. The standard definition in the literature is to say that a finite generalised quadrangle has \textit{order} $(s,t)$ if there are $s+1$ points incident with a given line and $t+1$ lines incident with a given point. Routine counting arguments can be used to establish the following well-known result.

\begin{lemma} \label{point-enum}
The number of points in a finite generalised quadrangle of order $(s,t)$ is $(s+1)(st+1)$.
For any point $x$ in the GQ, there are $s(t+1)$ points in $B_1(x)$ and $s^2t$ points in $B_2(x)$.
\end{lemma}

\begin{proof}
There are $t + 1$ lines through $x$, each containing $s$ points distinct from $x$.
Since a GQ contains no non-trivial triangles, these lines are disjoint (outside of $x$).
So there are $s(t+1)$ points collinear with $x$, and $|B_1(x)| = s(t+1)$.

Consider now a point $y$ in $B_2(x)$.  Since $y \notin B_1(x)$, $y$ is not incident
with any line through $x$; choose such a line $L$.  By the GQ Axiom, $y$ is collinear
with a unique point on $L$.  Since there are $s$ points on $L$ other than $x$, and each
of these points is collinear with $s(t+1) - s = st$ points not on $L$, there are exactly $s\cdot st = s^2t$ points in $B_2(x)$. 
\end{proof}

The following result, due to D.G. Higman, shows that the parameters $s$ and $t$ cannot differ by too much, in general.

\begin{theorem}[\cite{HigmanGQ}]\label{lem:GQst}
In a finite generalised quadrangle of order $(s,t)$, if $s>1$ and $t>1$, then $s \leqslant t^2$ and $t \le s^2$.
\end{theorem}

Our analysis of pseudonymity relations in a GQ-UPIR scheme will require the concept of a \textit{hyperbolic line} in a GQ, which we introduce now. In a finite generalised quadrangle $\mathcal{Q}$ of order $(s,t)$, given any two non-collinear points $x$ and $y$, by the GQ Axiom, there is a collection $\mathcal{C}$ of exactly $t+1$ points collinear with both $x$ and $y$.  Thus there are at least two points, $x$ and $y$, that are collinear with all the points in $\mathcal{C}$, but there could be more.

\begin{definition}\label{def:span}
Given a set of pairwise non-collinear points $\mathcal{X}$ in a finite generalised quadrangle, we define $B_1(\mathcal{X})$ to be the set of points collinear with each point in $\mathcal{X}$, \emph{i.e.}, $B_1(\mathcal{X}) = \bigcap_{x \in \mathcal{X}}\limits B_1(x)$.

We define the \textit{span} of $\mathcal{X}$ to be the set of points collinear with every point of $B_1(\mathcal{X})$, \emph{i.e.},
$\sp(\mathcal{X}) = B_1(B_1(\mathcal{X})) = \bigcap_{z \in \bigcap_{x \in \mathcal{X}}\limits B_1(x)}\limits B_1(z).$
\end{definition}

When $X = \{x_1, \dots, x_m\}$, we often write $B_1(x_1, \dots, x_m)$ to denote $B_1(\mathcal{X})$ and $\sp(x_1, \dots, x_m)$ to denote $\sp(\mathcal{X})$.  Note that, for non-collinear points $x$ and $y$ in a generalised quadrangle of order $(s,t)$ we have $\{x,y\} \subseteq \sp(x,y)$ and, by the GQ Axiom, $|B_1(x,y)| = t+1$.  The set $\sp(x,y)$ is often referred to as the \textit{hyperbolic line} defined by $x$ and $y$.  The following results show that hyperbolic lines have incidence properties similar to those of ordinary lines.

\begin{lemma}\label{lem:span1}
If $a \in \sp(x,y)$, then $\sp(a,x) = \sp(x,y)$.
\end{lemma}

\begin{proof}
Let $a \in\sp(x,y)$.  Because $a \in \sp(x,y) = B_1(B_1(x,y))$, $a$ is collinear with each point in $B_1(x,y)$.  Since $B_1(x,y) = B_1(a,x,y) \subseteq B_1(a,x)$, we have
\[t+1 = |B_1(x,y)| = |B_1(a,x,y)| \le |B_1(a,x)| = t+1.\]
Hence $B_1(x,y) = B_1(a,x)$, and therefore
$\sp(x,y) = B_1(B_1(x,y)) = B_1(B_1(a,x)) \\
= \sp(a,x).$
\end{proof}

\begin{corollary}\label{cor:span2}
If $|\sp(x,y) \cap \sp(w,z)| > 1$, then $\sp(x,y) = \sp(w,z)$.
\end{corollary}

\begin{proof}
Let $\{a,b\} \subseteq \sp(x,y) \cap \sp(w,z)$.  By Lemma \ref{lem:span1},
$\sp(x,y) = \sp(a,x) = \sp(a,b) = \sp(w,a) = \sp(w,z).$
\end{proof}

We end this section by collecting relevant information about the families of \textit{classical generalised quadrangles}.  These families are related to certain classical groups, and are thus highly symmetric.  In each case, the size of spans of sets of a given type are constant, regardless of which points within the generalised quadrangle are chosen.  In the following table, $q$ is a prime power, and $x$, $y$, and $z$ are three mutually noncollinear points.

\begin{table}[ht]\label{tab:hyplines}
\begin{center}
\caption{The classical generalised quadrangles.}\label{tbl:classical}
\begin{small}
	\setlength{\tabcolsep}{12pt}
\begin{tabular}{l|l|l}
\toprule
$\mathcal{Q}$& Order & Span size\\
\midrule
$\W(3,q)$, $q$ odd& $(q,q)$ & $|\sp(x,y)| = q + 1$\\
$\Q(4,q)$, $q$ even& $(q,q)$ & $|\sp(x,y)| = q + 1$\\
$\Q(4,q)$, $q$ odd & $(q,q)$ & $|\sp(x,y)| = 2$\\
$\Q^-(5,q)$ & $(q,q^2)$ & $|\sp(x,y,z)| = q + 1$\\
$\herm(3,q^2)$ & $(q^2,q)$ & $|\sp(x,y)| = q^2 + 1$\\
$\herm(4,q^2)$ &$(q^2,q^3)$ & $|\sp(x,y)| = q+1$\\
$\herm(4,q^2)^D$ &$(q^3,q^2)$ & $|\sp(x,y)| = 2$\\
\bottomrule
\end{tabular}
\end{small}
\end{center}\vspace{-0.7cm}
\end{table}

\section{Secure GQ-UPIR systems}

We begin this section by describing in detail the pseudonymity relation on a GQ-UPIR scheme for a single eavesdropper.

\begin{proposition}\label{distance2}
In a GQ-UPIR scheme using Protocol \ref{prot1}, the pseudonymity classes with respect to a single eavesdropper $c$ are singleton classes for users at distance $1$ from $c$, and are of the form $\sp\langle c, u\rangle \setminus \{c\}$ for any user $u$ at distance $2$ from $c$.
\end{proposition}

\begin{proof}
Suppose that $c$ observes a series of linked queries; if $c$ shares a message space with $u$, then the queries always appear in this message space. Otherwise by the GQ axiom, for every line through $c$, there is a unique user on that line collinear with $u$. This implies that $c$
observes linked queries from $u$ distributed uniformly across all message spaces to which $c$ has access. Hence $c$ can decide whether $u$
is at distance $1$ or distance $2$.

If $d(c,u) = 1$, then an argument exactly analogous to that of Theorem \ref{distance1} shows that $c$ can resolve the identity of $u$.
So suppose that $d(c,u) = 2$. For a fixed line $M$ containing $c$, there is a unique user $u_{1}$ sharing a message space with $u$. Over sufficiently many linked queries, $c$ will observe queries addressed to every user in $M$ \textbf{except} for $x$. As a result, $c$
learns $\mathcal{X} = B_{1}(c) \cap B_{1}(u)$.

Now recalling Definition \ref{def:span}, suppose that $v \in \sp\langle u, c_{1}\rangle$. Then by Lemma \ref{lem:span1}, $B_{1}(c_{1}) \cap B_{1}(v) = \mathcal{X}$. It follows that $u$ and $v$ are pseudonymous. Likewise, any other user in $\sp \langle u, c_{1} \rangle \setminus c_{1}$ falls into the same pseudonymity class.
\end{proof}

An easy corollary of Theorem \ref{distance2} is that a single user in a GQ-UPIR scheme can resolve the identity of every other user in the scheme
if and only if every hyperbolic line in the GQ has size $2$. There are two such families known: $Q(4,q)$ where $q$ is an odd prime power, and $H(4, q^{2})^{D}$. The data given in Table \ref{tab:hyplines} shows that the pseudonymity relation on a GQ-UPIR scheme will never give a giant component when using Protocol \ref{prot1}. We introduce a new protocol, inspired by onion routing, which will be secure against large coalitions of users.

\begin{protocol}\label{prot2}
	Let $(\mathcal{U} \cup \mathcal{M}, E)$ be a GQ-UPIR system.
    Suppose furthermore that every user has been assigned a public key, and that user $u$ wishes to retrieve the response to the query $Q$ from the database.
    \begin{small}
	\begin{enumerate}
		\item $u$ chooses a user $v$ uniformly at random from the set of all users, and generates a private key $\psi$.
		\item If $u = v$, $u$ requests $Q$ directly from the server, receiving response $R$.
        \item If $d(u, v) = 1$, then user $u$ encrypts both the query $Q$ and the private key $\psi$ using $v$'s public key $\phi_{v}$, and writes
        the request $[v, \phi_{v}(Q), \phi_{v}(\psi)]$ to the unique message space that they share.
        \item Otherwise, $d(u,v) = 2$ and $u$ chooses a shortest path to $v$, say $[u, M_{1}, u_{1}, M_{2}, v]$. Message passing is as in Protocol~\ref{prot1}: $u$ writes the query $[(u_{1}, M_{2}, v), \phi_{v}(Q), \phi_{v}(\psi)]$ to $M_{1}$.
        \item When $v$ receives the request, he forwards $Q$ to the database, receives response $R$ and writes the response $[(v), \phi_{v}(Q), \psi(R)]$ to the message space in which the query was observed. The response is returned to user $u$ as in Protocol \ref{prot1}.
	\end{enumerate}
\end{small}
\end{protocol}

\begin{remark}
In Protocol \ref{prot2}, the only user who learns the query $Q$ is the proxy $v$; this means that users do not observe linked queries addressed to other users. The use of a private key is necessary since revealing $u$'s public key to $v$ would compromise $u$'s privacy.
\end{remark}

\begin{proposition}\label{pseudo}
In a GQ-UPIR scheme using Protocol \ref{prot2},
all users at distance two from a single eavesdropper $c$
are pseudonymous.
\end{proposition}

\begin{proof}
As in Proposition \ref{distance2}, $c$ can identify whether the source $u$ of a series of linked queries is at distance $1$ or distance $2$.
Suppose that $u$ is at distance $2$. Then $c$ observes linked queries in every message space to which he has access with equal probability.
Without observing queries addressed to other users he gains no information about members of $B_{1}(u)$, and so any pair of users at distance
$2$ are pseudonymous.
\end{proof}

In contrast to GQ-UPIR schemes, it can be shown that encryption of messages offers limited benefits in
PBD-UPIR schemes. In particular, a single eavesdropper learns which message space he shares with another user: this means that there can be no giant component in the pseudonymity relation.
It can be shown that a coalition of three users, not all sharing a single message space, suffices to identify any other user in a projective plane UPIR scheme. In essence, intersection attacks do not require coalition members to observe queries addressed to other users.

\begin{theorem}
A GQ-UPIR scheme using Protocol $2$ is secure against coalitions of users of size $O(n^{1/4-\epsilon})$.
Hence any family of GQ-UPIR schemes is secure in the sense of Definition \ref{defn:priv}.
\end{theorem}

\begin{proof}
Let $c_{1}$ and $c_{2}$ be members of a coalition, and suppose that $d(c_{1}, u) = d(c_{2}, u) = 1$.
Write $M_{i}$ for the (unique) message space shared by $c_{i}$ and $u$. An intersection attack
identifies $u$ as the unique user with access to both spaces. We conclude that the identity of
any user at distance $1$ from more than $1$ coalition member can be resolved.

Suppose now that $c_{1}$ is the unique coalition member at distance $1$ from $u$. Then $c_{1}$ learns
the message space $M$ shared with $u$. By the GQ axiom, every other user is at distance $1$
from a unique user in $M\setminus \{u\}$. So the coalition gains information about the identity of $u$
proportional to the number of neighbours that they have in $M$. In particular, the identity of $u$ is
completely resolved if and only if every user in $M\setminus \{u, c_{1}\}$ is a neighbour of some coalition member.
In this last case, the size of the coalition is at least proportional to the number $s+1$ of points on a line in the GQ.

By Proposition \ref{pseudo}, the users at distance $2$ from every member of the coalition $C$ form a
single pseudonymity class with respect to $C$. We estimate the size of this pseudonymity class. By Lemma
\ref{point-enum}, the number of users at distance $1$ from a single user is $s(t+1)$. To apply Definition
\ref{defn:priv}, we require $|C| (st+s) \leq n^{1-\epsilon}$ where $n = s^{2}t + st + s + 1$ is the number of
users in the scheme.

Using Lemma \ref{lem:GQst}, we write $t = s^{\alpha}$ where $\frac{1}{2} \leq \alpha \leq 2$. Straightforward
calculations show that $n = O(s^{2+\alpha})$ and that $st + s = O(s^{1+\alpha})$ from which we conclude that
$ st + s = O(n^{\frac{1+\alpha}{2+\alpha}}) = O(n^{3/4}),$
where the second equality comes from $\alpha \leq 2$. It follows that when $|C| = O(n^{1/4 - \epsilon})$,
the set of users at distance $2$ from $C$ form a giant component in the pseudonymity relation of $C$. 
\end{proof}

\section*{Acknowledgements}

The authors acknowledge the assistance of John Bamberg with questions concerning generalised quadrangles.
This research was partially supported by the Academy of Finland (grants \#276031, \#282938, \#303819, \#283262, and \#283437).

\bibliographystyle{abbrv}
\flushleft{
\bibliography{NewBiblio}

\def\Dbar{\leavevmode\lower.6ex\hbox to 0pt{\hskip-.23ex \accent"16\hss}D}
\begin{thebibliography}{10}

\bibitem{BIKR02}
A.~Beimel, Y.~Ishai, E.~Kushilevitz, and J.~F. Raymond.
\newblock Breaking the {$O(n^{1/(2k-1)})$} barrier for information-theoretic
  private information retrieval.
\newblock In {\em Proceedings, {FoCS} 2002}, pages 261--270, 2002.

\bibitem{BJL}
T.~Beth, D.~Jungnickel, and H.~Lenz.
\newblock {\em Design theory. {V}ol. {I}}, volume~69 of {\em Encyclopedia of
  Mathematics and its Applications}.
\newblock Cambridge University Press, Cambridge, second edition, 1999.

\bibitem{CKGS98}
B.~Chor, E.~Kushilevitz, O.~Goldreich, and M.~Sudan.
\newblock Private information retrieval.
\newblock {\em J. ACM}, 45(6):965--981, Nov. 1998.

\bibitem{Domingo-Ferrer}
J.~Domingo-Ferrer, M.~Bras-Amor\'{o}s, Q.~Wu, and J.~Manj\'{o}n.
\newblock User-private information retrieval based on a peer-to-peer community.
\newblock {\em Data Knowl. Eng.}, 68(11):1237--1252, Nov. 2009.

\bibitem{DG16}
Z.~Dvir and S.~Gopi.
\newblock 2-{S}erver {PIR} with subpolynomial communication.
\newblock {\em J. ACM}, 63(4):39:1--39:15, Sept. 2016.

\bibitem{HigmanGQ}
D.~G. Higman.
\newblock Invariant relations, coherent configurations and generalized
  polygons.
\newblock In {\em Combinatorics, {P}art 3: {C}ombinatorial group theory}, pages
  27--43. Math. Centre Tracts, No. 57. Math. Centrum, Amsterdam, 1974.

\bibitem{PayneThas}
S.~E. Payne and J.~A. Thas.
\newblock {\em Finite generalized quadrangles}, volume 110 of {\em Research
  Notes in Mathematics}.
\newblock Pitman, Boston, MA, 1984.

\bibitem{StokesBrasAmores}
K.~Stokes and M.~Bras-Amor{\'o}s.
\newblock Optimal configurations for peer-to-peer user-private information
  retrieval.
\newblock {\em Comput. Math. Appl.}, 59(4):1568--1577, 2010.

\bibitem{SwansonStinsonI}
C.~M. Swanson and D.~R. Stinson.
\newblock Extended combinatorial constructions for peer-to-peer user-private
  information retrieval.
\newblock {\em Adv. Math. Commun.}, 6(4):479--497, 2012.

\bibitem{SwansonStinsonII}
C.~M. Swanson and D.~R. Stinson.
\newblock Extended results on privacy against coalitions of users in
  user-private information retrieval protocols.
\newblock {\em Cryptogr. Commun.}, 7(4):415--437, 2015.

\bibitem{GRS97}
P.~F. Syverson, M.~G. Reed, and D.~M. Goldschlag.
\newblock Private web browsing.
\newblock {\em J. Comput. Secur.}, 5(3):237--248, June 1997.

\end{thebibliography}
}
		
\end{document}